\documentclass[conference]{IEEEtran}
\usepackage{amsmath}
\usepackage{amsfonts}
\usepackage{amssymb}
\usepackage{graphicx}
\usepackage{textcomp}
\usepackage{multirow}
\usepackage{lineno}
\usepackage[linesnumbered,ruled]{algorithm2e}
\usepackage{amsthm}
\usepackage{amsopn}
\usepackage{lineno}
\usepackage{color}
\usepackage{wasysym}
\DeclareMathOperator*{\argmax}{arg\,max}
\ifCLASSINFOpdf
\else
\fi
\begin{document}
\title{A Novel Distributed Pseudo TDMA Channel Access Protocol for Multi-Transmit-Receive Wireless Mesh Networks}
\author{\IEEEauthorblockN{Yuanhuizi~Xu, Kwan-Wu~Chin
\IEEEauthorblockA{School of Electrical, Computer and Telecommunications Engineering \\
University of Wollongong, NSW, Australia\\
Email: yx879@uowmail.edu.au, kwanwu@uow.edu.au}}
\and
\IEEEauthorblockA{Sieteng~Soh\\
Department of Computing\\
Curtin University, WA, Australia\\
Email: s.soh@curtin.edu.au}
}
\maketitle

\begin{abstract}
Wireless Mesh Networks (WMNs) technology has been used in recent years for broadband access in both cities and rural areas. A key development is to equip routers with multiple directional antennas so that these routers can transmit to, or receive from multiple neighbors simultaneously. The Multi-Transmit-Receive (MTR) feature can boost network capacity significantly if suitable scheduling policy is applied. In this paper, we propose a distributed link scheduler called PCP-TDMA that fully utilizes the MTR capability. In particular, it activates every link at least once within the shortest period of time. We evaluated the performance of PCP-TDMA in various network topologies, and compared it against a centralized algorithm called ALGO-2, and two distributed approaches: JazzyMAC and ROMA. The results show that PCP-TDMA achieves similar performance with the centralized algorithm in all scenarios, and outperforms the distributed approaches significantly. Specifically, in a fully connected network, the resulting superframe length of PCP-TDMA is less than $^1/_3$ and $^1/_2$ of JazzyMAC and ROMA, respectively. 
\end{abstract}
\IEEEpeerreviewmaketitle

\section{\label{IN}Introduction}
Wireless Mesh Networks (WMNs) are developing quickly and gaining a lot of attention because of their ubiquitous applications. Recently, researchers have discovered the Multi-Transmit-Receive (MTR) capability of WMNs when mesh routers are equipped with multiple radios\cite{raman2P}. This capability allows nodes to transmit to \textit{or} receive from multiple neighbors over the same frequency simultaneously without causing collisions; see Figure \ref{2P}(a)(b). As a result, MTR WMNs have a much higher network capacity than conventional WMNs that use an omni-directional \cite{localgreedyjoo} \cite{maximizingmodiano} or a single directional antenna \cite{kwanSDMA} \cite{patrawildnet} \cite{duttakchannel}. However, nodes are half-duplex, meaning they adhere to the no Mix-Tx-Rx constraint; see Figure \ref{2P}(c).

\begin{figure}[htbp]
\centering
\includegraphics[width=3.5in]{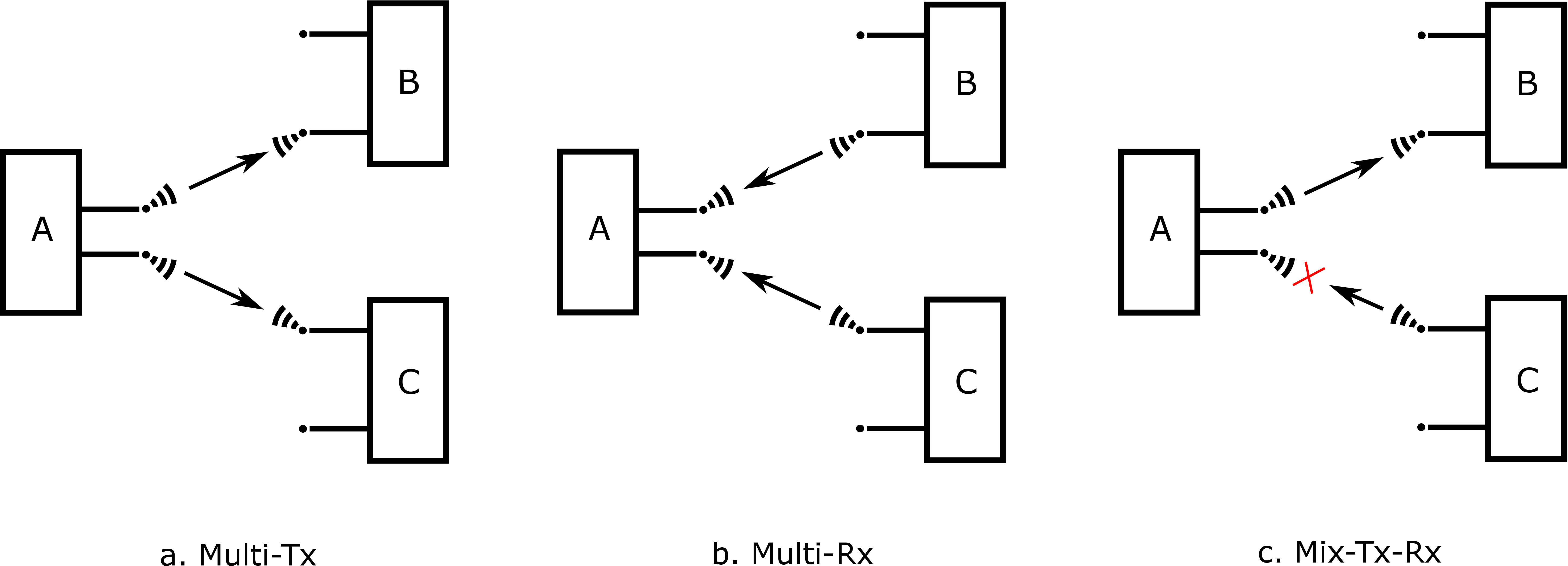}
\caption{MTR capability of node $A$: (a) transmissions, (b) receptions, and the key constraint (c) cannot transmit \textit{and} receive at the same time.}
\label{2P}
\end{figure}

For MTR WMNs to achieve maximum network capacity, the fundamental problem is to design a link scheduler, subject to the Mix-Tx-Rx constraint, ensures the maximum number of links are activated in each time slot. This problem is proved to be a NP-complete, MAX CUT problem in a prior work \cite{kwanmaxcut}. The authors of \cite{kwanmaxcut} provided one solution by using brute force, which is possible for small networks. For larger network, brute force becomes computationally intensive. Currently, the existing polynomial time approaches, such as \cite{raman2P}, \cite{kwanCTR} and \cite{kwanmaxcut}, are primarily centralized which requires a central controller to monitor the number of connection of each node. Based on this information, the controller controls the nodes to transmit or receive in each time slot. In order to collect information from every node, a centralized policy will incur many rounds of signaling overheads, propagation and contention delays. Motivated by these limitations of centralized solutions, we design a simple distributed algorithm called Period Controlled Pseudo-TDMA (PCP-TDMA) to gradually derive the MAX CUT in order to determine the shortest superframe for any MTR networks over time.

To illustrate how PCP-TDMA solves the problem, consider Figure \ref{probEG}. We see that the initial superframe $\mathcal{SF}_a$ has a length of four slots.  All links are activated as per the no mix-Tx-Rx constraint. Over time, we see that the links, e.g., $e_{BC}$ and $e_{CB}$, adapt their transmission slot with the aim of reducing the superframe length, i.e., allocating $e_{BC}$ in slot 2 of $\mathcal{SF}_f$ and $e_{CB}$ in slot 1 of $\mathcal{SF}_k$ shortens the superframe length to 2. From this example, we see that an important goal is to adapt the length of the superframe, which requires nodes to move their allocated slots closer to the start of a superframe. Thus, in this paper, our main contributions:
\begin{figure}[htbp] 
	\centering
	\includegraphics[width=3.4in]{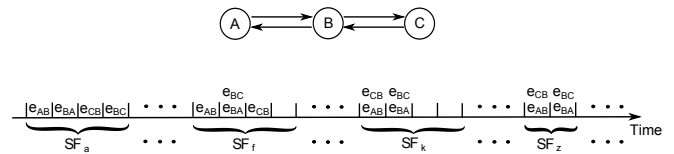}
	\caption{Example topology and schedule}
	\label{probEG}
\end{figure}

\begin{itemize}
\item We establish the network model of MTR WMNs where nodes are equipped with multiple directional antennas and all nodes are able to concurrently transmit or receive on all links. 
\item We propose PCP-TDMA, a simple distributed link scheduler that generates max cut in MTR WMNs. From simulation results, PCP-TDMA achieves similar performance as the centralized algorithm ALGO-2 \cite{kwanmaxcut} in all network scenarios.  
\item As compared to past distributed solutions, e.g., JazzyMAC \cite{nedevJazzyMAC} and ROMA \cite{ROMAbao}, PCP-TDMA provides three advantages over prior methods. First, each node only communicates with its one-hop neighbors, and it does not require any global topological information. Second, it achieves high fairness because each link is guaranteed to activate at least once in a superframe. Third, the superframe generated by PCP-TDMA is shorter than the other distributed solutions.
\item We analyse and prove that PCP-TDMA produces interference-free schedules for arbitrary topologies, and nodes will converge in a limited time.
\end{itemize}

\section{\label{PRE}Preliminaries}
We model a WMN as a connected graph $G (V, E)$ with $|V|$ vertices and $|E|$ edges. Each node $v \in V$ represents a static wireless mesh node, and each edge $e_{uv}$ or $(u, v)$ in $E$ corresponds to a directed link from node $u$ to $v$ in $G$ if and only if the Euclidean distance between $u$ and $v$ is smaller than or equal to the transmission range $r$. Here, each link is supported by a radio and each node $u$ has $b_u \ge |N(u)|$ radios, where $N(u)$ contains the neighbors of node $u$.

Time is divided into slots of equal length, which are sized accordingly to transmit one packet.  Nodes are assumed to be synchronized \cite{GPSsync}; e.g., using GPS.  The superframe is denoted as $\mathcal{SF}$ and consists of up to $P$ edge sets, where $P$ is the superframe length; aka the period.  We define the $i$-th edge set in $\mathcal{SF}$ as $\epsilon_i$, which contains transmitting links that adhere to the no mix-Tx-Rx constraint; note, $\epsilon_i$ can be empty. Hence, a superframe is defined as $\mathcal{SF}=\{\epsilon_i\; |\; i \in \{1, \ldots,P\}\}$. In each superframe, every time slot $s$ is numbered sequentially, whereby $s_i$ represents the $i$-th slot with $i \in \{1, \ldots, P\}$. In addition, we will index the $x$-th superframe as $\mathcal{SF}_x$. 

To avoid interference, nodes need to know the slots that their neighbors use for transmitting and receiving packets.  To this end, all transmitting and receiving slots of a node $v$ are included in a set called $\textbf{Tslot}_v$ and $\textbf{Rslot}_v$, respectively. As an example, consider $\textbf{Tslot}_{A} = \{s_i, s_j\}$. This means node $A$ is scheduled to transmit in slot $s_i$ and $s_j$. How this set is established will be detailed in Section \ref{part1}. We assume each node knows the $\textbf{Tslot}_v$ and $\textbf{Rslot}_v$ of every neighbor $v$. This is reasonable because each node can include this information in all transmitted data packets. As a result, when selecting a transmitting slot for its link, say $(A, B)$, node $A$ can only choose from the set of {\em feasible slots}, which is defined as $\mathcal{S}_{(A, B)} = \{s_1, \ldots, s_P\} \setminus (\textbf{Rslot}_{A} \cup \textbf{Tslot}_{B})$.   In words, we exclude slots used for reception and those used by node B for transmission.  Table \ref{notation} summarizes key notations used throughout this paper.
\begin{table}[htbp]
	\centering
	\begin{tabular}{ |l|l| }
		\hline
		{\em Notation} & {\em Description} \\ \hline
		$G$ & A directed graph \\
		$V$ & A set of nodes or vertices in $G$ \\
		$E$ & A set of directed links or edges in $G$ \\
		$e_{AB}$ or $(A, B)$ & A link with source node $A$ and destination node $B$ in $E$ \\
		$N(u)$ & A set containing node $u$'s neighbors \\
		$\mathcal{SF}_x$ & The $x$-th superframe \\
		$P$ & The length of $\mathcal{SF}$, \textit{aka} the period or superframe length \\
		$s_i$ & The $i$-th slot in $\mathcal{SF}$ \\
		$\textbf{Tslot}_A/\textbf{Rslot}_A$ & A set that contains node $A$'s transmitting/receiving slots \\
		$\rho$ & The probability that a node attempts to reserve a slot again \\
		$T_{O}$ & The duration of a timeout timer started by {\em parent} node \\
		$\mathcal{S}_{(A, B)}$ & A set containing link $(A, B)$'s {\em feasible slots} \\
                          $D_{max}$ & The maximum node degree among all nodes\\
                          $\diameter$ & The network diameter\\
		\hline
	\end{tabular}
	\caption {Key notations}
	\label{notation}
\end{table}

We are now ready to define the problem. Our aim is to derive the shortest possible superframe; i.e., the smallest $P$, in a distributed manner.  That is, given an initial $P$ value, and nodes with MTR capability, design a distributed algorithm that iteratively reduces the superframe length or $P$ value over time. Note, how $P$ is determined initially and adjusted will be discussed in Section \ref{ANA} and \ref{ALGO}, respectively.

\section{\label{ALGO} Period Controlled Pseudo TDMA (PCP-TDMA)}
The basic idea is as follows.  The initial superframe $\mathcal{SF}_1$ has period $P$.  All nodes attempt to reserve a random slot for each of their links.  If a node reserves a slot successfully, it will use the slot for data transmission in the next superframe. Otherwise, if there is a collision, a node will attempt to reserve another random slot in the next superframe. After reserving a slot, say $s_i$, the node will then attempt to improve its current slot $s_i$ by reserving an earlier slot $s_j$, where $j <i$, in the next superframe. This means if the last slot reserved by nodes is $s_c$, where $c<P$, then we can reduce the superframe to $c$; i.e., we update $P$ to $c$. 

Consider Figure \ref{example}.  Assume $|\mathcal{SF}_1|=P=6$. In the first $P$ slots, none of the links are scheduled; i.e., $\mathcal{SF}_1$ consists of six empty edge sets.    Nodes send a RESV message for each of their links; this is indicated by the gray boxes.  Assume the RESV message of all links, except $e_{CB}$, is delivered successfully. Transmission on link $e_{CB}$ fails because $B$'s reception is affected by the transmission on link $e_{BA}$.  From the next superframe onwards, the transmitter of these links will start to transmit data packets, shown as white boxes, in their reserved slot. This means these links are included in the edge sets of $\mathcal{SF}_2$. Additionally, as soon as a node, say $v$, successfully reserves a slot, it marks this slot as its transmitting slot and includes this slot into $\textbf{Tslot}_v$. On the other hand, if a node, say $w$, receives a RESV message in slot $k$, it includes $k$ into $\textbf{Rslot}_w$. 
\begin{figure}[htbp]
	\centering
	\includegraphics[width=3.4in]{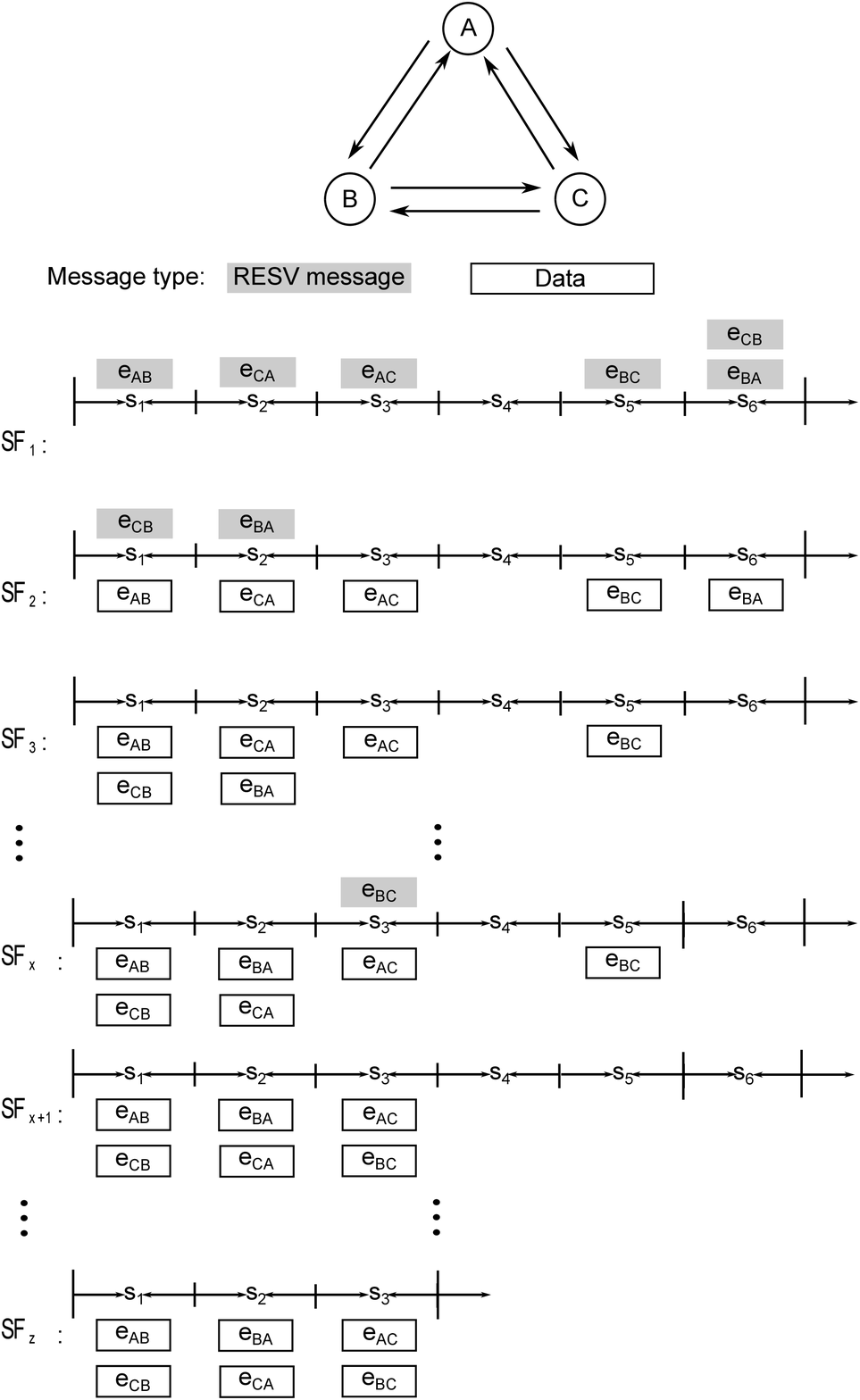}
	\caption{An example of PCP-TDMA}
	\label{example}
\end{figure}
 
As mentioned earlier, node $B$ fails to receive the RESV message over link $e_{CB}$. Consequently, in superframe $\mathcal{SF}_2$, node $C$ sends another RESV message in a random slot in the set $\mathcal{S}_{(C, B)}$, say slot $s_1$, for the unscheduled link $e_{CB}$. As link $e_{CB}$ does not interfere with any link in the set $\epsilon_1$, it is thus included in $\epsilon_1$ of superframe $\mathcal{SF}_3$. 

Nodes can improve the slot of the reserved links by contending for an earlier slot.  Continuing the previous example, we see that link $e_{BA}$ occupies slot $s_6$.  In order to improve $s_6$, node $B$ sends a RESV message for slot $s_2$; note, this slot is chosen randomly from the set $\mathcal{S}_{(B, A)}$. As a result, link $e_{BA}$ is removed from the edge set $\epsilon_6$ and added into the edge set $\epsilon_2$ in superframe $\mathcal{SF}_3$. In addition, nodes $B$ and $A$ will update their set $\textbf{Tslot}_{B}$ and $\textbf{Rslot}_{A}$, respectively, by replacing element $s_6$ with $s_2$. Next, in $\mathcal{SF}_{x}$, node $B$ tries again to reserve an earlier slot, i.e., $s_3$, for link $e_{BC}$. We see that $e_{BC}$ is then scheduled in $s_3$ of superframe $\mathcal{SF}_{x+1}$. 

The next key idea of PCP-TDMA is to reduce the superframe length or $P$ iteratively.  Assume node $B$ has the highest node ID among the three nodes, as node $B$ has reached a state where no link can be shifted to an earlier slot without causing interference, node $B$ is prompted to propose a new period. Assume nodes know the slots reserved by its neighbors, then node $B$ searches for the latest time slot used by itself and its neighbors; i.e., $s_3$ of $\mathcal{SF}_{x+1}$ is used by link $e_{BC}$. Thus node $B$ proposes $P=3$. If all nodes approve this new period, meaning slots after $s_3$ are idle, all nodes set $P$ to three. 

Finally, we see that the resulting schedule converges to $\mathcal{SF}_z$ with a length of three. Formally, we define the {\em converged state} as follows: 
\theoremstyle{definition}
\newtheorem{PA}{Definition}
\begin{PA}
\label{defconv}
The system reaches the {\em converged state} when all nodes are unable to improve the current reserved slot of all their links. 
\end{PA}
Referring to Definition \ref{defconv}, the superframe $\mathcal{SF}_z$ and all subsequent superframes will have a period of three. Note, if the topology changes, i.e., a new node joins, then a new superframe will have to be regenerated. However, we do not expect this to happen frequently as nodes are primarily static in a WMN.  Having said that, in Section \ref{ALGO}, we will discuss how a new node can be incorporated into an existing WMN. 

In the foregone example, we see that PCP-TDMA needs to address the following four sub-problems. Firstly, nodes need to improve their reserved slots over time. Secondly, it is necessary for nodes to reduce the probability of collisions when reserving a random slot. Thirdly, nodes need to determine the last reserved slot. Fourthly, it is important that {\em all} nodes update their period to the same value and reduce to the shortest possible $P$.

PCP-TDMA consists of two parts: slot reservation and period minimization. In the first part, nodes send RESV messages to move the activation time of their links nearer to the start of each superframe in order to fully utilize the time slots at the front side of a superframe. In the second part, nodes communicate with their neighbors to inform each other about the idle slots located at the end of a superframe, and then remove these idle slots to shorten the period $P$. 

\subsection{\label{part1} Part-1: Slot Reservation}

In this part, nodes aim to improve their current transmission slots by attempting to reserve earlier slots. Figure \ref{fpart1} shows the state diagram of the slot reservation process.  Initially, nodes are in the ``\textit{Start}'' state. Assume link $e_{AB}$ of node $A$ currently has slot $s_i$.  Node $A$ will attempt to reserve a random slot $s_j$ in $\mathcal{S}_{(A, B)}$, where $j<i$, by sending a RESV message to node $B$ in slot $s_j$.  The RESV message is sent with a probability of $\rho=\frac{i}{P}$; recall that  $i$ is the slot index number, and $P$ is the current superframe length.  Observe that $\rho$ is biased towards links with a bigger slot number; i.e., those near the end of the current superframe will have a higher priority to move to an earlier non-conflicting slot. 

\begin{figure}[htbp]
\centering
\includegraphics[width=3in]{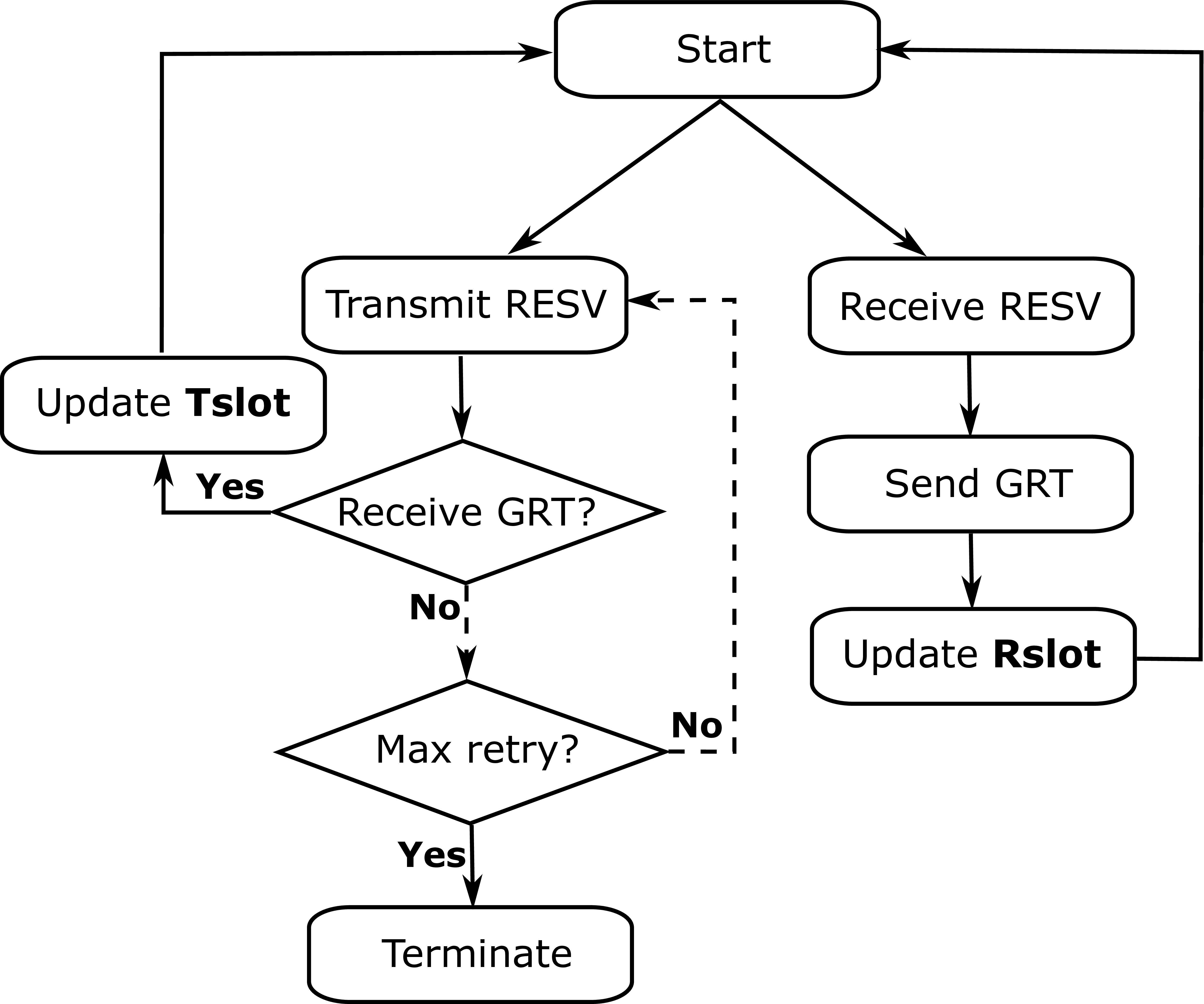}
\caption{State diagram for PCP-TDMA's slot reservation process}
\label{fpart1}
\end{figure}

A node, say $B$, that receives a RESV message moves into the state ``\textit{Receive RESV}''.   Assume node $B$ receives a RESV message without any conflict.  It then replies immediately with a grant or GRT message. Node $B$ then updates $\textbf{Rslot}$  to record slot $s_j$ as its receiving slot; i.e., it replaces slot $s_i$ in $\textbf{Rslot}_{B}$ with $s_j$.  After that node $B$ goes back to the ``\textit{Start}'' state. When node $A$ receives the GRT message from $B$, it moves to the ``\textit{Update $\textbf{Tslot}$}'' state to mark slot $s_j$ as its transmitting slot by replacing slot $s_i$ with $s_j$ in the set $\textbf{Tslot}_{A}$. It then moves back to the ``\textit{Start}'' state. 

If node $B$ experiences a collision, i.e., it did not receive the RESV message from $A$, then there will be no GRT message.  In this case, node $A$ concludes that the reservation has failed. It thus retains the current transmitting slot $s_i$ for link $e_{AB}$.  Node $A$ will either go back to the ``\textit{Transmit RESV}'' state to retransmit a RESV message with probability $\rho$ in a random slot from $\mathcal{S}_{(A, B)}$ in the next superframe, or go to the ``\textit{Terminate}'' state. The state node $A$ chooses depends on whether it has tried to transmit a RESV message for a given maximum retry threshold.  We set the retry limit to $|\mathcal{S}_{(A, B)} \setminus \{s_i, \ldots, s_P\}|$, where $s_i$ is the current reserved slot. This allows nodes to try to reserve in every earlier slot in $\mathcal{S}_{(A, B)}$ before it terminates the slot reservation process. Once in the ``\textit{Terminate}'' state, a node no longer tries to move its current slots.

\subsection{\label{part2}Part-2: Period Minimization}
This part consists of two stages: new $P$ proposal and its confirmation. The aim is for nodes to learn the shortest feasible period and to update their current period. Eventually, all nodes in the network will use the same shortest $P$, and the superframe period can no longer be shortened. 

\subsubsection{\label{stage1}Stage-1: New $P$ Proposal}

We first explain how nodes propose a new period. To reduce signalling overheads, only nodes with the highest ID among all their neighbors have the right to propose.   Assume that node $A$ is such a node.  After reaching the ``\textit{Terminate}'' state in Part-1, it searches for the largest slot that is occupied by a transmitting link.  Formally, we have
\begin{equation}
	\label{P'}
            	P' = \argmax_{k\in \{1,\ldots,P\}} (s_k\cap\{\textbf{Tslot}_{u} \cup \textbf{Rslot}_{u}\} \neq \emptyset)
\end{equation}
where $u \in \{N(A) \cup A\}$.  Node $A$ compares $P'$ against the current period $P$. If $P'< P$, then node $A$ becomes the root node. It sends a PROP\{$A$, $P'$\} message to its neighbors. The message includes its ID and the proposed period $P'$.

In the sequel, we will need the following definition of {\em parent} and {\em child}.  A {\em parent} of a node $A$ is defined as the neighbour that has transmitted a PROP message to $A$. All other nodes in $N(A)$ are known as the {\em children} of node $A$.  For each PROP message, a node will keep a separate record of the corresponding parent and child nodes. In addition, after transmitting a PROP to every child, a node starts a timer called $T_{O}$. The duration of $T_{O}$ is a design parameter that can be changed according to traffic requirements or network topology. 

When a node, say $C$, receives a PROP \{$A$, $P'$\} message from its neighbor $B$, node $C$ needs to determine whether to accept or reject the proposed period $P'$. This process is illustrated by the state diagram shown in Figure \ref{part2stage1}. Upon receiving a PROP message, node $C$ will record neighbor $B$ as a {\em parent}. Then node $C$ needs to determine whether it is a duplicated PROP message.  To do this, node $C$ checks the following two elements contained in the PROP message: ID and $P'$. If the ID of the received PROP message matches the ID contained in a previously received PROP message, and these two PROP messages have the same $P'$ value, then the newly received PROP message is a duplicate. Node $C$ discards the duplicated PROP message and will not reply to parent $B$.

\begin{figure}[htbp]
\centering
\includegraphics[width=3in]{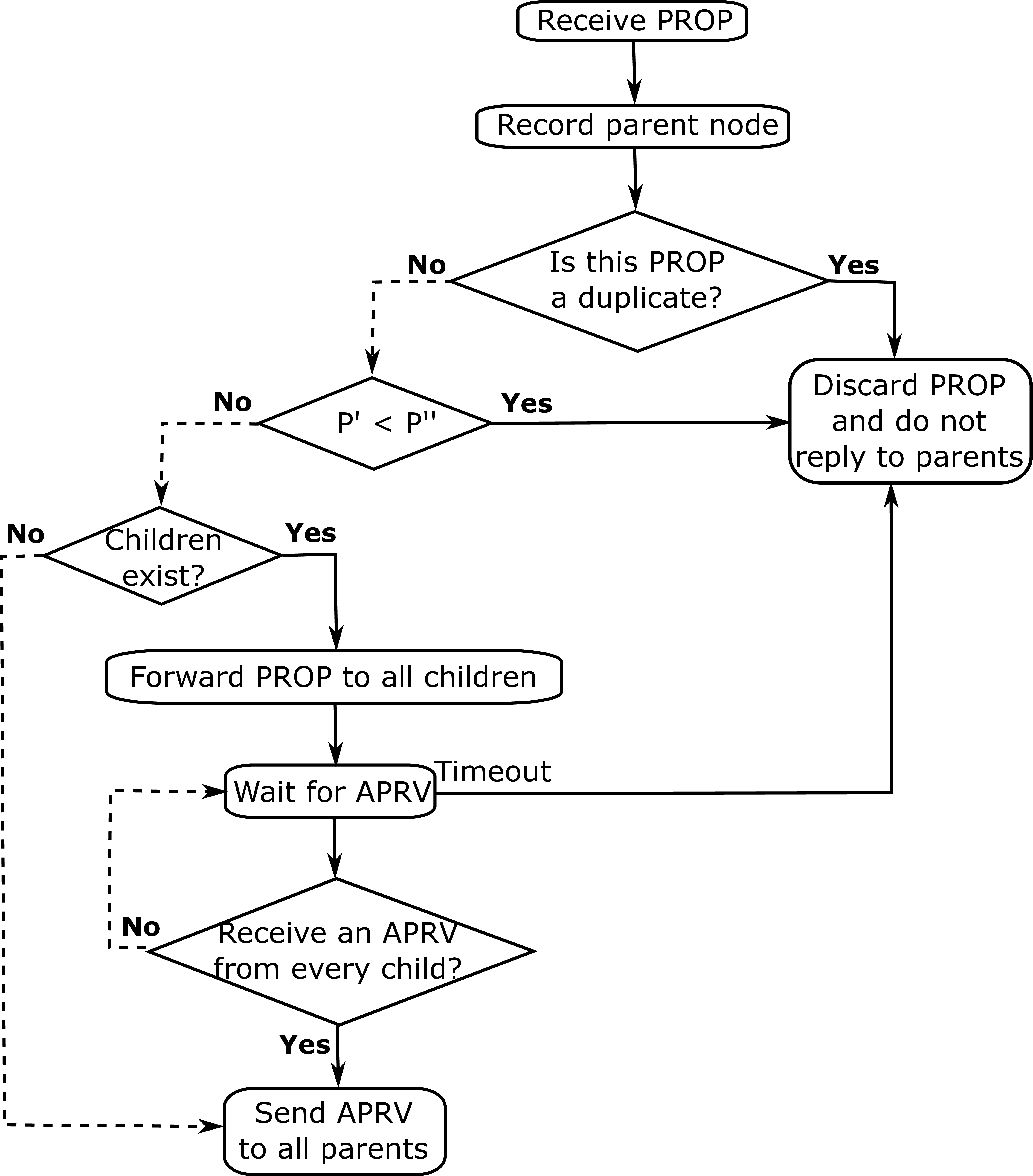}
\caption{The propagation of a PROP message}
\label{part2stage1}
\end{figure}

On the other hand, if the PROP message is new, then $C$ will determine the validity of the proposed period $P'$ as follows. Node $C$ first finds its largest occupied time slot $s_l'$.  That is, 
\begin{equation}
	\label{P''}
	P''=\argmax_{k\in\{1,\ldots,P\}}(s_k\cap\{\textbf{Tslot}_{u}\cup\textbf{Rslot}_{u}\} \neq \emptyset)
\end{equation}
where $u \in \{N(C) \cup C\}$. Node $C$ compares $P'$ against $P''$ to determine whether the proposed period $P'$ can be accepted by node $C$. If $P'$ is smaller than $P''$, meaning $P'$ cannot be node $C$'s new period, then node $C$ discards the message PROP \{$A$, $P'$\} and does not reply to any {\em parent}. On the contrary, if $P' \ge P''$, node $C$ approves $P'$ and forwards the PROP message to all its {\em children}, if there are any. Node $C$ then enters the state ``\textit{Wait for APRV}'', where $C$ expects all its {\em children} to send a message APRV\{$A$, $P'$\} back as an approval of the proposed period $P'$. If $C$ does not receive an APRV message from every child within a duration of $T_{O}$, a  Timeout  event occurs. This causes node $C$ to discard the message PROP \{$A$, $P'$\} and not to reply with APRV. However, if $C$ collected every APRV before a Timeout event, then node $C$ goes to the last state ``\textit{Send APRV to all parents}''.

\subsubsection{\label{stage2}Stage-2: New $P$ Confirmation}

A root node, say $A$, that has successfully collected an APRV message from all its neighbors (or children) in Stage-1 proceeds to Stage-2. The aim of this stage is to inform all nodes the approved period and the start time of the new superframe that has this period. The key challenge is to have all nodes start the new superframe at the same time. 

In this stage, any node, say $C$, uses a message called UPDATE \{$A$, $P'$, $t$, $\tau_{C}$\} to inform its children that the root node $A$ is going to start a superframe with period $P'$. Here, the element $t$ is a time stamp (e.g., unix epoch timestamp) of when this UPDATE message is generated by the root node, and $\tau_{C}$ indicates the time slot that node $C$ begins using the new  $P'$ instead of the current period $P$. Here, all four elements are important because they are also used to guarantee the uniqueness of each UPDATE message.

Now we explain how node $C$ calculates its starting slot $\tau_{C}$. Assume node $C$ will send the UPDATE \{$A$, $P'$, $t$, $\tau_{C}$\} message to its children in $\mathcal{SF}_{y}$. The value of $\tau_{C}$ must satisfy the following two requirements. 
\begin{enumerate}
	\item {\bf Requirement-1}: $\tau_{C}$ must be a slot in $\mathcal{SF}_{y+2}$. This is because node $C$ requires one superframe $\mathcal{SF}_{y}$ to send an UPDATE message to all children and another superframe $\mathcal{SF}_{y+1}$ to receive an ACK from all its children. 
	
	\item {\bf Requirement-2}: $\tau_{C}$ must be $n \times P'$ slots after the starting slot of $C$'s parent, where $n \in \mathbb{N}$. This is to ensure that $C$ starts the new superframe with period $P'$ simultaneously with its parent. Here, the new superframe $\mathcal{SF}_{y+2} = \{\epsilon_i\; |\; i \in \{1, \ldots,P'\}\}$, where each edge set $\epsilon_i$ in $\mathcal{SF}_{y+2}$ is equal to $\epsilon_i$ in $\mathcal{SF}_{y}$.
\end{enumerate}

We now use Figure \ref{part2stage2} to explain how a new period is confirmed and updated by every node. Firstly, a root node, say $A$, will carry out the steps on the left branch. It sends the message UPDATE\{$A$, $P'$, $t$, $\tau_{A}$\} to all its children and waits for their ACK. At time slot $\tau_{A}$, if $A$ has received an ACK from every child, node $A$ goes to the last state in the left branch; i.e., ``\textit{Start new superframe with period of $P'$ in slot $\tau_{A}$}''. Otherwise, node $A$ returns to the sending UPDATE state at the beginning of the left branch after ``\textit{Recalculate $\tau_{A}$}''. Here, the starting slot $\tau_{A}$ is recalculated as the first slot after two superframes; cf. Requirement-1.

\begin{figure}[htbp]
	\centering
	\includegraphics[width=3.4in]{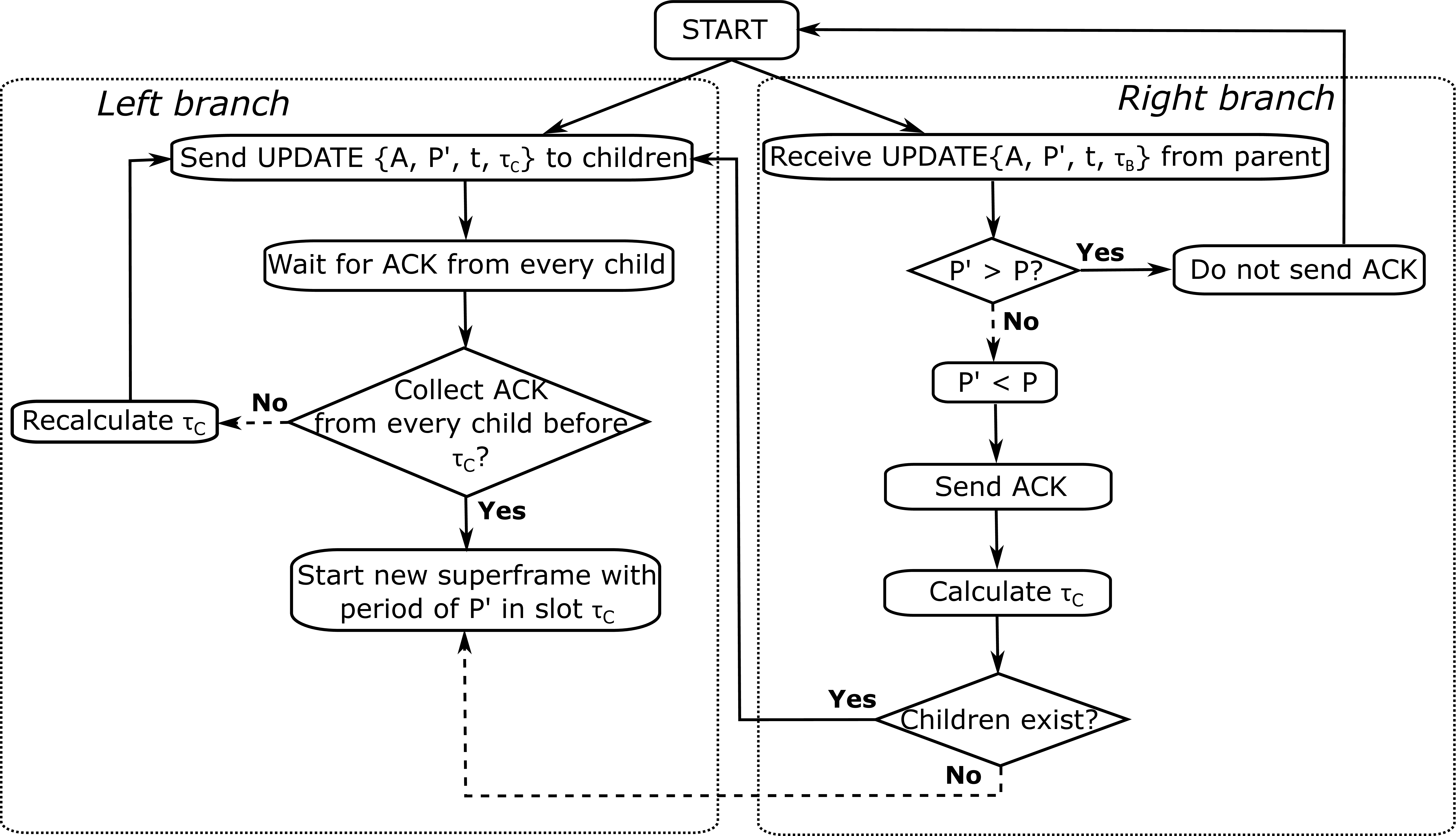}
	\caption{Confirmation of a new period}
	\label{part2stage2}
\end{figure}
 
Next, we describe how a node confirms and updates the new period $P'$ when it receives an UPDATE message. According to the right branch of Figure \ref{part2stage2}, if a node, say $C$, receives an UPDATE \{$A$, $P'$, $t$, $\tau_{B}$\} message from its parent $B$, node $C$ acquires the following information: parent $B$ is going to start a new superframe with period $P'$ from slot $\tau_{B}$ onwards.  
	
First, node $C$ compares $P'$ against the current period $P$. If $P'$ is bigger than $P$; i.e., $P' > P$, node $C$ does not send an ACK to its parent and it will not change its period. Otherwise, if $P' < P$, node $C$ goes through the remaining states in the right branch. It first responds to parent $B$ with an ACK to inform $B$ that the period update event has been noted. Then node $C$ will calculate its own starting slot $\tau_{C}$ for the new period. If node $C$ has children, it moves to the left branch to inform its children. Otherwise, if $C$ has no children, it goes to the last state in the left branch where $C$ starts the new superframe with period of $P'$ from $\tau_{C}$ onwards. With the propagation of the UPDATE messages, the period update event occurs at each node in the network. Finally, all nodes conform to the same period $P'$.

One question that arises is that what if node $C$ receives an UPDATE message whose $P'$ equals to the current period $P$. Although superframe length does not change, a different starting slot contained in the newly received UPDATE message leads to a different starting slot $\tau_{C}$ for node $C$. Thus, to ensure that every node starts the new superframe simultaneously, node $C$ will check the time stamp $t$ of this UPDATE message. Nodes will adopt this UPDATE message if it is older because this means it has existed for a longer time period and thus, covers more nodes. In the event that the time stamp is the same, nodes will accept the UPDATE if it contains a higher root node ID. Otherwise, the UPDATE message will be discarded.

\subsubsection{\label{newnodes}New Nodes}
Whenever new nodes join, the current period may need to be re-adjusted to ensure that new links can be scheduled without interference. In addition, we also need to make sure that existing links remain unaffected by the introduction of new links.

Assume there is a new node $F$.  Let $E$ be its neighbour that is already connected to the network.  Node $E$ will transmit its current schedule and current period $P$ to node $F$ in a random slot, say $s_r$, in $\textbf{Tslot}_E$. Node $F$ will then record $s_r$ in $\textbf{Rslot}_F$. Upon receiving node $E$'s schedule, node $F$ inspects the schedule of $E$, and sends a RESV message to node $E$ in a random {\em feasible} slot $s_t$ in $\mathcal{S}_{(F, E)}$. If node $F$ receives a GRT message from $E$, then the slot reservation is successful. Node $F$ and $E$ add the reserved slot $s_t$ into the set $\textbf{Tslot}_F$ and $\textbf{Rslot}_E$, respectively.

In the case where no {\em feasible} slots are available, meaning every slot has a conflict with node $F$'s outgoing links, then node $F$ needs to expand the current superframe by one slot. Specifically, the new superframe needs to have a period of $P+1$, where the edge sets $\{\epsilon_i | i \in \{1, \ldots,P\}$ in the new superframe are equal to that of the current superframe, and the one extra edge set $\epsilon_{P+1}$ contains all $F$'s unscheduled links. 

To expand the superframe, node $F$ sends a JOIN\{$P+1$, $\tau_0$\} message to an existing  neighbor node at random.  Assume this neighbour to be $E$.  It then becomes a root node and sends a EXP\{$E$, $P+1$, $t$, $\tau_{E}$\} message to its children, where $E$ is the root node ID, $P+1$ is the new period, $t$ is the time stamp and $\tau_{E}$ is the starting slot of the new superframe. This message is propagated to all other nodes in the network; see Figure \ref{newnode}.  Observe that the depicted process is similar to how an UPDATE message is processed in Section \ref{stage2}, except that nodes do not have to verify the validity of the proposed period $P+1$.

\begin{figure}[htbp]
	\centering
	\includegraphics[width=3.4in]{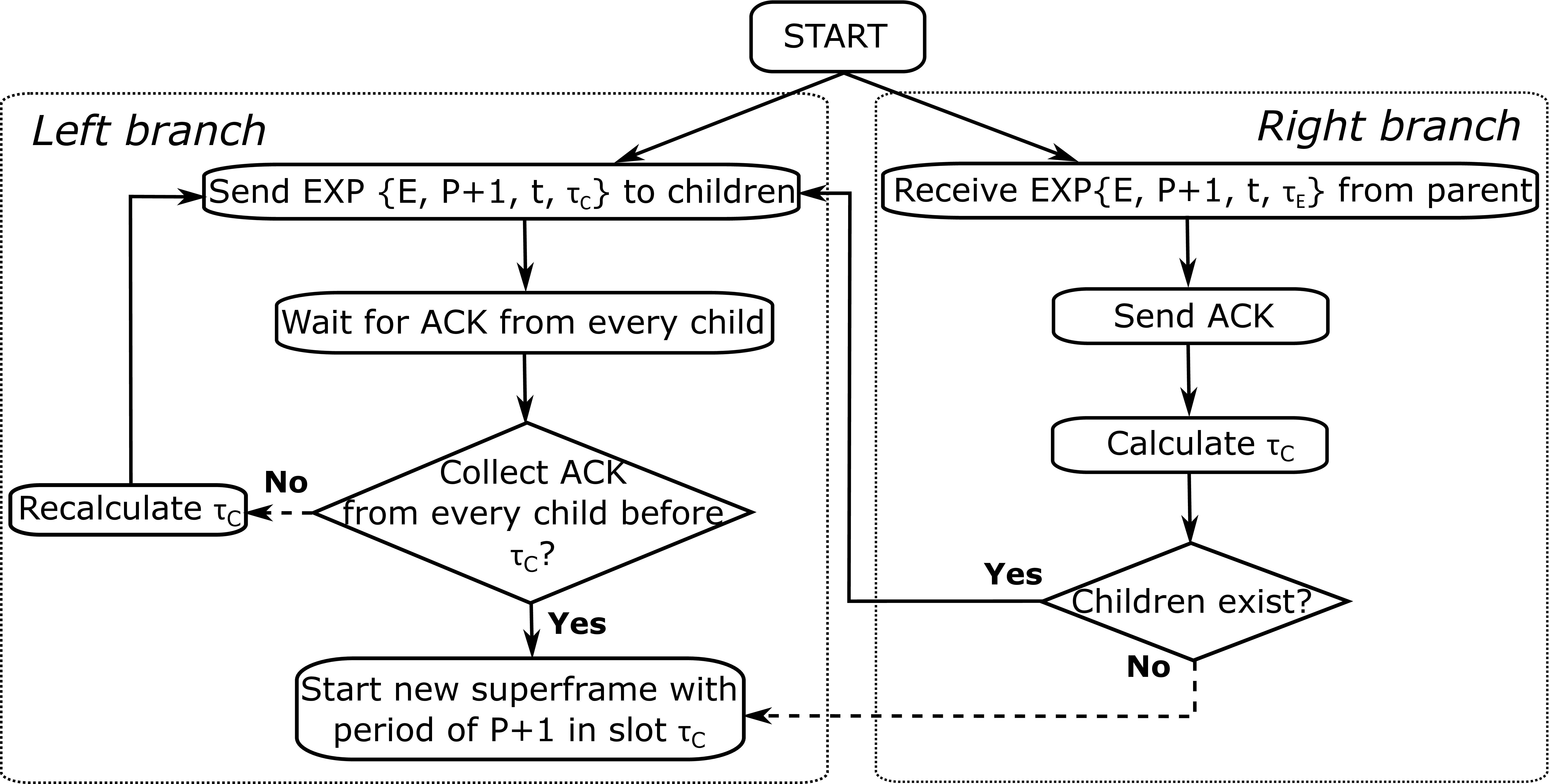}
	\caption{Re-adjusting the period when a new node joins}
	\label{newnode}
\end{figure}

\section{\label{ANA}Analysis}
In this section, we analyze several properties of PCP-TDMA, including the configuration of the initial period, the correctness of the schedule, the self-stabilizing feature of the algorithm, and the time required for Part-2 of PCP-TDMA to finish. 
\newtheorem{proposition}{Proposition}
\begin{proposition}
\label{iniP} 
Given an arbitrary topology, with a maximum node degree $D_{max}$, setting the initial period $P_i$ to at least $2\times D_{max}$ guarantees each link will reserve a slot.
\end{proposition}
\begin{proof}
Consider a link $(A, B)$.   To schedule this link without interference, the following inequality must be true:
\begin{equation}
	\mathcal{S}_{(A, B)} \neq \emptyset
\end{equation}
Equivalently, 
\begin{equation}
	\{s_1, \ldots, s_P\} \setminus (\textbf{Rslot}_{A} \cup \textbf{Tslot}_{B})\neq \emptyset
\end{equation}
This indicates that the number of feasible slots must be greater than zero. Since the values are all integers, we have the following inequality,  
\begin{equation}
	P-|\textbf{Rslot}_{A} \cup \textbf{Tslot}_{B}| \geq 1
\end{equation}
If both node $A$ and $B$ have $D_{max}$ neighbors, that means $A$ has $D_{max}$ incoming links,  and $B$ has $D_{max}$ outgoing links. In the worst case, all these said links are scheduled in a distinct time slot. Consequently, we have $|\textbf{Rslot}_{A}|=|\textbf{Tslot}_{B}|=D_{max}$. Note, link $(B, A)$ is counted twice. Thus, we have
\begin{equation}
	P-(2\times D_{max}-1) \geq 1
\end{equation}
Hence, we obtain $P \geq 2\times D_{max}$ ensures that all links have at least one feasible slot to reserve, which proves the proposition. 
\end{proof}
\begin{proposition} 
PCP-TDMA produces an interference-free schedule.
\end{proposition}
\begin{proof}
We are only interested in Part-1 (\textbf{Slot Reservation}) of PCP-TDMA because Part-2 (\textbf{Period Minimization}) reduces the length of the superframe without changing the link schedule.  In Part-1, consider a node $A$ and assume links $e_{AB}$ and $e_{xA}$ have reserved the same slot $s_i$. We have two cases to consider:

{\em Case 1}:  Node $A$ transmits a RESV message in slot $s_i$ even though a link $e_{xA}$ exist. Recall that $A$ can only select a transmitting slot from the set $\mathcal{S}_{(A, B)}$. The fact that slot $s_i$ is in both $\mathcal{S}_{(A, B)}$ and $\textbf{Rslot}_{A}$ contradicts the definition of  a {\em feasible slots} set, whereby $\mathcal{S}_{(A, B)} = \{s_1, \ldots, s_P\} \setminus (\textbf{Rslot}_{A} \cup \textbf{Tslot}_{B})$.  

{\em Case 2}: Node $A$ and a neighbor $B$ choose to send a RESV message in time slot $s_i$. As the reservation is successfully, this means node $A$ receives the RESV message from node $B$ while $A$ is sending its RESV message to $B$. This contradicts the no Mix-Tx-Rx constraint. 

In both cases, PCP-TDMA does not generate a schedule with interference, which proves the proposition.
\end{proof}
\begin{proposition}
PCP-TDMA has the property of self-stabilization, which ensures all nodes in the network to end up in a correct state; i.e. the (\em converged state).
\end{proposition}
\begin{proof}
We show that nodes using PCP-TDMA reach the {\em converged state}. In Part-1, a node, say $A$, iteratively attempts to replace the current reserved slot $s_i$ with a random slot $s_j$ from the set $\mathcal{S}_{(A, B)}$ for its link $e_{AB}$, where $j<i$. If the attempt is successful, the maximum retry limit is updated to $|\mathcal{S}_{(A, B)} \setminus \{s_j,\ldots, s_P\}|$. Otherwise, if this attempt fails, the retry limit becomes $|\mathcal{S}_{(A, B)} \setminus \{s_i,\ldots, s_P\}|-1$ because $s_j$ is removed from $\mathcal{S}_{(A, B)}$. Thus, the max retry threshold is guaranteed to decrease to zero at some time, meaning node $A$ will eventually move to the ``\textit{Terminate}'' state. Note, this ``\textit{Terminate}'' state is equivalent to the {\em converged state} because nodes no longer change their transmitting and receiving slots. Therefore, PCP-TDMA is self-stabilizing because all nodes are guaranteed to reach the {\em converged state}.
\end{proof}
\begin{proposition} 
The number of slots, denoted as $\sigma$, required by Part-2 of PCP-TDMA in an arbitrary network with diameter $\diameter$ is bounded by $2\times P \leq \sigma \leq 4\times \diameter \times P$.
\end{proposition}
\begin{proof}
We first consider Stage-1 of Part-2.  We bound the number of superframes a root node requires to receive an APRV from every neighbor after initiating a PROP message. In the best case, this can be done in only one superframe if the transmission of PROP happens successively from root node to the farthest node, and from parents to children. Then within the same superframe, after a PROP message is received by the farthest node, the transmission of APRV messages occurs in the exact opposite sequence of PROP's transmission, i.e., from the farthest node to root node, children to parents. However, without this specific transmission order, it may take up to at most $ \diameter \times P$ slots to propagate a PROP message to the farthest node from the root node and another $ \diameter \times P$ slots for the root node to collect all APRV messages. This happens when the hop-distance between root and the farthest node equals the network diameter $\diameter$. Thus, the number of time slots PCP-TDMA takes to perform Stage-1 of Part-2 is at least $P$, and at most $ 2\times\diameter\times P$. Similarly, we have the same results for Stage-2 of Part-2 for the transmission of UPDATE and ACK messages. Therefore, the number of slots required by Part-2 is bounded by $[2\times P, 4\times \diameter\times P]$.
\end{proof}

\section{\label{SIM}Evaluation}
We evaluate the performance of PCP-TDMA using MatGraph \cite{matgraph}, a Matlab toolkit that works with simple graphs.   Each node is assumed to have a dedicated antenna for every neighbor.    We conduct our experiments over bipartite or random topologies. For experiments that use bipartite graphs, we construct a linear and a grid network consisting of 16 nodes. For random topologies, we place 50 nodes randomly on a $100m \times 100m$ square area in order to study the impact of two parameters: node degree and transmission range. The degree of each node varies from 5 to 15.   We vary the transmission range of nodes from 30m to 100m. 

We compare PCP-TDMA against Algo-2 \cite{kwanmaxcut}, a centralized MTR link scheduler, and two distributed algorithms JazzyMAC \cite{nedevJazzyMAC}, and ROMA \cite{ROMAbao}.  The aim of Algo-2 is to generate a bipartite graph with maximal matching by placing nodes into two sets: Set1 and Set2. Initially, all nodes are included in Set1 while Set2 is empty.  Algo-2 then moves a node from Set1 to Set2 if doing so increases the number of active links. After processing all nodes, a max cut is derived. In time slot $i$, nodes in Set1 transmit to nodes in Set2. Then, upon removing all activated links from nodes in Set1 to those in Set2 from the network, the above process is repeated on the revised topology to generate the next max cut.  Algo-2 terminates when it has scheduled all links.   The superframe length is equal to the total number of max cuts obtained by Algo-2. This is the minimal superframe length that ensures every link is activated at least once.  

JazzyMAC initially assigns tokens to nodes according to a centralized scheme; i.e., graph coloring.  A node becomes a transmitter when it holds the token of all its incident links. When a node finishes its transmission, it passes the token to the other end of the link.  ROMA is a distributed scheme where nodes are synchronized and uses two-hop topology information to compute a schedule.  ROMA evenly and randomly splits nodes into transmitters and receivers in each slot, which are paired together for data transmission. Then ROMA solves any contention according to the priority of each node, where the priority is calculated based on node ID. The node with the highest priority among contending neighbors has the right to transmit without conflicts in that time slot.

In our experiments, we compare metrics such as superframe length and the number of concurrent active links.  In addition, we also measure the number of time slots and signaling messages required for PCP-TDMA to reach convergence. All presented results are an average of 20 simulation runs; each with a different topology. The error bars shown in the line graphs indicate 95\% confidence interval of the mean value. 

\subsection{\label{dresults}Node Degree}
Figure \ref{degree} (a) shows the superframe length calculated when nodes have 5 to 15 neighbors. We can see that all the algorithms generate a relatively short superframe with similar length except ROMA.  The superframe of ROMA is approximately two times more than that of other algorithms. This is because ROMA splits all nodes into transmitters and receivers randomly in each time slot.  However, the other three algorithms construct a max cut comprising of unscheduled links, and thus they schedule the maximal number of unscheduled links in each slot which leads to a shorter superframe.  Interestingly, when nodes have a degree of seven, we observe that JazzyMAC generates a superframe with length that outperforms the centralized algorithm ALGO-2 by one.  The reason is because in the initial greedy graph coloring stage of JazzyMAC, it occasionally generates the optimal graph coloring. However, ALGO-2 in these cases fails to derive the optimal max cut. 

\begin{figure}[htbp] 
	\centering
	\includegraphics[width=3.4in]{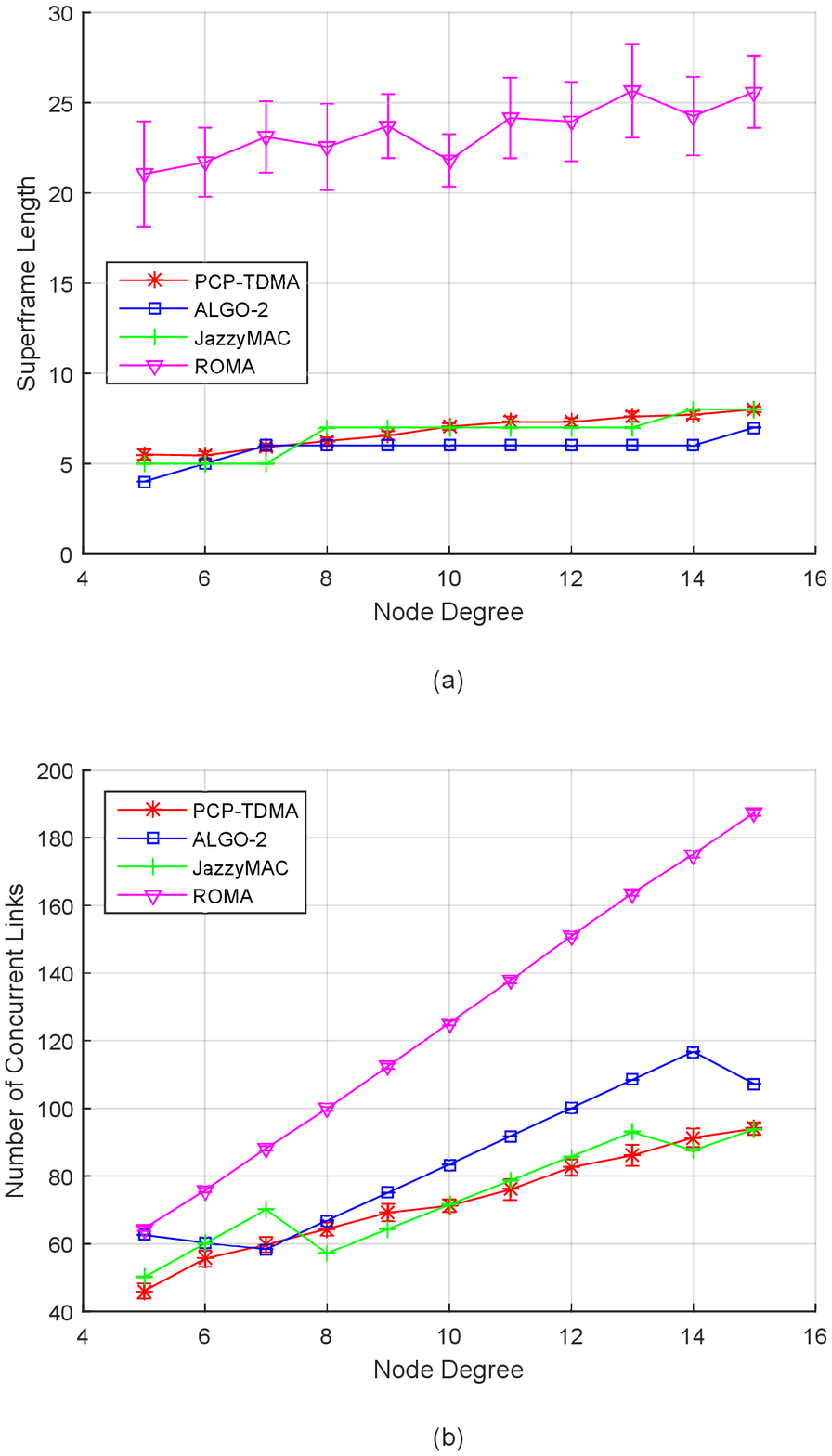}
	\caption{Performance of different algorithms under increasing node degrees. (a) Superframe length. (b) Number of concurrent links.}
	\label{degree}
\end{figure}
 
Figure \ref{degree} (b) shows the average number of concurrent links in each slot with increasing node degrees.  When using ROMA, the number of concurrent links increased from 64.4 to 187.3.  Specifically, it significantly outperforms other tested algorithms when node degree increases from 6 to 15. This is because ROMA does not remove any links after links are scheduled. As we increase the node degree, the number of existing links increases. Hence, ROMA has more chances to repeatedly schedule previously activated links as opportunistic links. However, for PCP-TDMA, we do not schedule opportunistic links because we want to fill every slot with the most number of unscheduled links. For ALGO-2, scheduled links are intentionally removed from the network. In JazzyMAC, opportunistic links do not exist because of its token scheme. As a result, these three algorithms have poorer performance in terms of the number of activated links in each slot. 

We also note that for PCP-TDMA, ALGO-2 and JazzyMAC, the product of superframe length and number of concurrent links per slot equals $|E|$. With each increment in node degree, $|E|$ increases by 50. Thus, we can see that if the superframe length does not increase, the number of concurrent links rises linearly. For example, when the node degree increases from seven to 14, the superframe length of ALGO-2 in Figure \ref{degree} (a) remains at six, while the number of concurrent links of ALGO-2 in Figure \ref{degree}(b) increases linearly. The increment value 8.3 is the result of $\frac{50}{6}$, where 50 is the number of added links, and six is the number of slots in a superframe. Moreover, we notice that when the superframe length increases by one or more, the number of concurrent links decreases. For instance, the superframe length of JazzyMAC increases from five to seven in Figure \ref{degree} (a) when the node degree increases from seven to eight. With more slots in a superframe, from Figure \ref{degree}(b), there will be fewer links in each slot on average.   

\subsection{\label{trresults}Transmission Range}
Next, we conduct experiments on networks with 50 nodes; these nodes are randomly located on a $100 \times 100 m^2$ area.  We vary the transmission range from 30 to 100 meter. 
\begin{figure}[htbp] 
	\centering
	\includegraphics[width=3.4in]{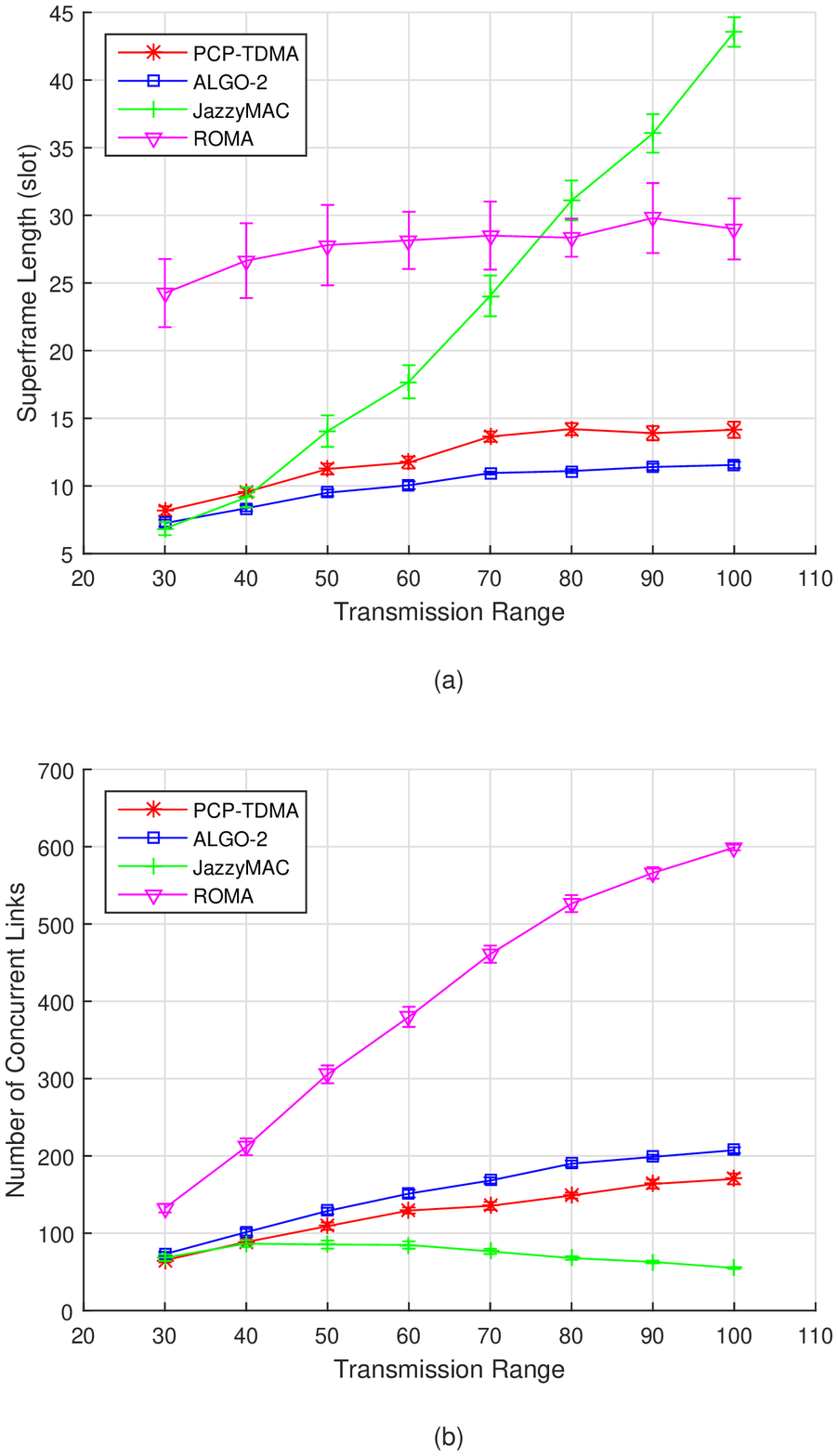}
	\caption{Performance of different algorithms under increasing transmission range, (a) Superframe length, and (b) Number of concurrent links.}
	\label{transrange}
\end{figure}
 
Figure \ref{transrange} (a) shows that ALGO-2 generates the shortest superframe length, which gradually increases from 6.2 to 11.5. The key reason for this increase is because more links are established between nodes as the transmission range increases. The superframe length of PCP-TDMA is close to that of ALGO-2; i.e., PCP-TDMA produces superframes with at most 3.1 additional slots.  For ROMA, its superframe length is fairly high at around 27. This is due to the same reason explained in Section \ref{dresults} where ROMA schedules links by randomly splitting nodes. The superframe length of JazzyMAC is similar with ALGO-2 and PCP-TDMA when the transmission range is 30 to 40m. However, from 40m onwards, JazzyMAC shows a sharp increase in superframe length. This is because in JazzyMAC a node is allowed to transmit on all its links only after it has the token of all its links. Consequently, in some cases, time slots are wasted while waiting for tokens to return. Thus, JazzyMAC's performance degrades when nodes need to collect more tokens from more neighbors.

Figure \ref{transrange} (b) compares the average number of concurrent links per time slot.  ROMA results in the most concurrent links because of opportunistic links. This value increases from 132.5 to 598.9 because of the growth in the total number of links. The number of concurrent links of when using ALGO-2 and PCP-TDMA doubles when the transmission range reaches 100m from 30m. For longer transmission ranges, the difference between ALGO-2 and PCP-TDMA is at most 20\%. For JazzyMAC, the number of concurrent links reduced by half when the transmission range increases from 40 to 100m. Thus, JazzyMAC is not suitable for random topologies when nodes have many neighbors. Note, when the transmission range reaches 100m, the network becomes almost fully connected. Thus, all results remain the same from 100m onwards.

\subsection{\label{bp}Bipartite Graphs}
 
Figure \ref{bipgraph} compares the superframe length generated by different algorithms for bipartite networks such as line and grid topology. We can see that both ALGO-2 and JazzyMAC have the shortest superframe length; i.e., two slots. This indicates that every node is acting as a transmitter in one time slot and as a receiver in the next slot; see  Figure \ref{bipexample} (a) for an example, where the number next to links indicates the $x$-th time slot of one superframe. The reason for the shorter superframe is because ALGO-2 constructs max cuts and JazzyMAC applies optimal graph coloring during bootup. PCP-TDMA yields superframe close to four slots. This is because in PCP-TDMA, links are scheduled in a random order. Take Figure \ref{bipexample} (b) as an example.  If link $e_{AB}$, $e_{BA}$, $e_{CD}$ and $e_{DC}$ are scheduled first as per the indicated slot number, then link $e_{CB}$ and $e_{BC}$ require two additional slots for interference-free transmission. In conclusion, the superframe of PCP-TDMA has an upper bound of four slots for bipartite graphs. 

\begin{figure}[htbp] 
	\centering
	\includegraphics[width=3.0in]{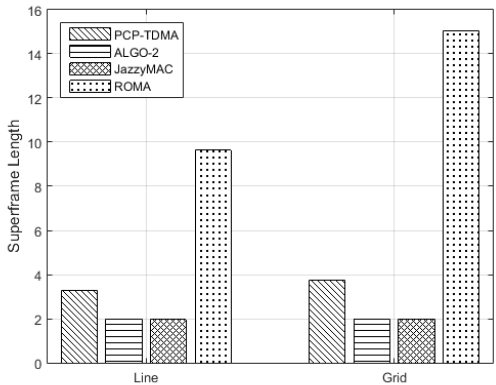}
	\caption{Superframe length for bipartite networks}
	\label{bipgraph}
\end{figure}
\begin{figure}[htbp] 
	\centering
	\includegraphics[width=3in]{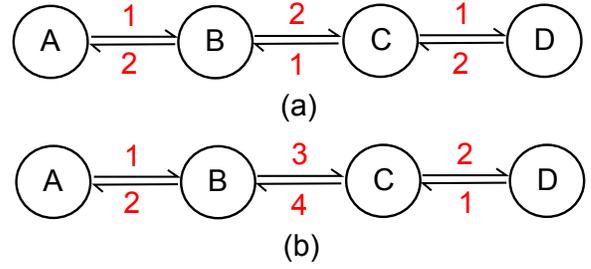}
	\caption{Example schedules for a line topology}
	\label{bipexample}
\end{figure}
         
\subsection{\label{convergence time}The impact of initial period on convergence time}
\begin{figure}[htbp] 
	\centering
	\includegraphics[width=3.4in]{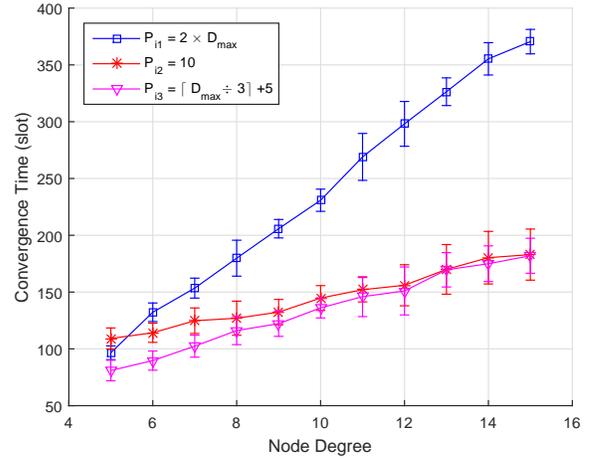}
	\caption{Convergence time under increasing node degrees}
	\label{convtime}
\end{figure}
 
In this section, we study that how does the initial period value $P_i$, i.e., the length of superframe $\mathcal{SF}_1$, affects the convergence time of PCP-TDMA. To do this, we compare the convergence time when we set the initial period $P_i$ to three different values. Figure \ref{convtime} illustrates the average number of time slots required for PCP-TDMA to reach convergence when using the following initial periods: $P_{i1}$, $P_{i2}$ and $P_{i3}$. Here, $P_{i1}$ is equal to $2 \times D_{max} $. This is the upper bound of the number of slots required to schedule every link; see Proposition \ref{iniP} in Section \ref{ANA}. We then set $P_{i2}$ to a constant of 10, and $P_{i3}=\lceil D_{max}/3 \rceil +5$. We perform this simulation on a 50-node network, with node degree increasing from five to 15. Overall, we see a rising trend in convergence time as node degree increases. This is because with increasing links, PCP-TDMA requires longer time to schedule every link. 

Next, we compare the convergence time when using different $P_{i}$ values. The three figures started at a similar value, around 100 slots. However, the convergence time when using $P_{i1}$ then rises significantly to 370 slots, whereas the number of slots when using $P_{i2}$ and $P_{i3}$ rose steadily to reach just 182. The reason is that, with the increase in node degree, the difference between the $P_{i1}$ and the final period $P_f$ increases rapidly, where $P_f$ is the length of superframes used by nodes when convergence is reached. This means when links are scheduled initially, they tend to be randomly scattered in a longer superframe. Thus nodes require more time to improve their reserved slots repeatedly, in order to reduce the superframe length from $P_{i1}$ to $P_f$.  Using $P_{i2}$ as the initial value results in the minimum increase, about 70 slots. The reason is that when node degree goes up, the difference between $P_{i2}$ and $P_f$ decreases as $P_{i2}$ is a constant. However, using a fixed integer as $P_i$ is not practical because $P_f$ increases proportionally to the maximum node degree. This positive correlation between $P_f$ and the maximum degree can be seen in Section \ref{dresults}. This means that if we set $P_i$ to a smaller value than $P_f$, PCP-TDMA can never compute a superframe because interference between links always exists. Thus we must ensure that $P_i$ is greater than $P_f$. From these results, we configure $P_i$ to be $P_{i3}$, which ensures a relatively small and constant difference from $P_f$. We can see in Figure \ref{convtime}, among the three $P_i$s, the convergence time when using $P_i3$ is the shortest, from 80 to 182 time slots.    

\subsection{\label{signalling}The number of signalling messages}
\begin{figure}[htbp] 
	\centering
	\includegraphics[width=3.4in]{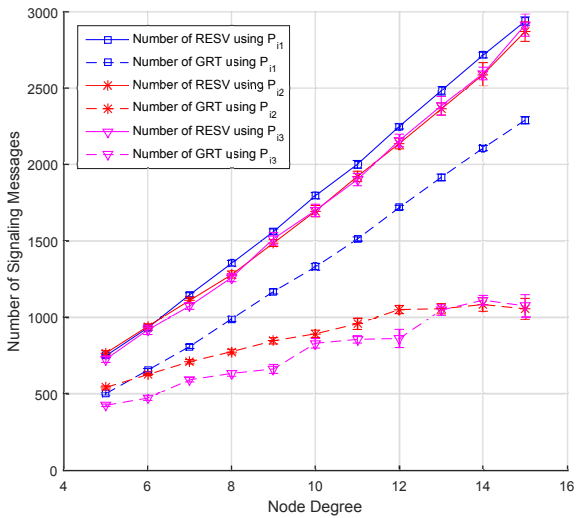}
	\caption{Total number of RESV and GRT messages transmitted to reach {\em converged state}}
	\label{numsigs}
\end{figure}
 
In this section, we study the number of signaling messages, including RESV and GRT, used by all nodes to reach {\em converged state}. The network configuration is the same as Section \ref{convergence time}. Figure \ref{numsigs} compares the total number of message exchanges incurred by PCP-TDMA to stabilize the schedule for all links when using defferent initial period values with increasing node degree.  We see that the number of GRT messages when using $P_{i1}$ is significantly higher than using $P_{i2}$ and $P_{i3}$; in fact, up to 50\% more. This is because $P_{i1}$ is greater than the other two $P_i$ values. With a longer initial superframe, there is a higher successful rate of reserving a random slot. For the same reason, the GRT messages when using $P_{i2}$ is also more than that of $P_{i3}$when the node degree is five to 12. Note, from 13 node degrees onwards, these two curves overlap because $P_{i2}=P_{i3}$ when degree is 13 to 15. In addition, we also notice that the number of GRT message when using $P_{i3}$ shows a step shape because $P_{i3}$ is a staircase function. 

On the other hand, we find that the numbers of RESV messages when using $P_{i1}$, $P_{i2}$ and $P_{i3}$ are very close. The reason is that, although nodes using $P_{i1}$ can easily reserve a slot for their links using fewer RESV messages as compared to using $P_{i2}$ and $P_{i3}$, they require more RESVs to improve reserved slots. Interstingly, the number of RESV messages rises linearly with increasing node degree. This indicates that nodes reserve 3.5 times on average for each of their links to allocate every link in the ideal slot, which is the earliest feasible slot for the particular link. This value does not increase with increasing node degree.  

\bibliographystyle{IEEEtran}
\bibliography{reference}
\end{document}